\documentclass[english,letterpaper,11pt,reqno]{amsart}
\usepackage{appendix}
\usepackage{amsbsy}
\usepackage{scalerel}
\usepackage{tcolorbox}
\usepackage{amsfonts}
\usepackage{amsmath}
\usepackage{amssymb}
\usepackage{amsthm}
\usepackage{graphicx}
\usepackage{ifthen}
\usepackage{textcomp}
\usepackage{enumitem, color, amssymb}
\usepackage{dsfont}
\usepackage{mathrsfs}
\usepackage{calrsfs}
\usepackage[bookmarksnumbered,colorlinks]{hyperref}
\hypersetup{colorlinks=true, linkcolor=blue, citecolor=blue}
\usepackage{url}
\usepackage{mathtools,relsize}
\usepackage{cancel}
\usepackage{setspace}

\newcommand{\RR}{\mathbb{R}}

\makeatletter
\newcommand*{\defeq}{\mathrel{\rlap{%
                     \raisebox{0.25ex}{$\m@th\cdot$}}%
                     \raisebox{-0.25ex}{$\m@th\cdot$}}%
                     =}
\makeatother

\makeatletter
\newcommand*\owedge{\mathpalette\@owedge\relax}
\newcommand*\@owedge[1]{%
  \mathbin{%
    \ooalign{%
      $#1\m@th\bigcirc$\cr
      \hidewidth$#1\m@th\wedge$\hidewidth\cr
    }%
  }%
}
\makeatother

\newtheorem{thm}{Theorem}
\newtheorem{lemma}{Lemma}
\newtheorem{cor}{Corollary}

\newtheorem{defn}{Definition}
\newtheorem{prop}{Proposition}
\newtheorem*{definition-non}{Definition}
\newtheorem*{theorem-non}{Theorem}
\newtheorem*{proposition-non}{Proposition}
\newtheorem*{lemma-non}{Lemma}
\newtheorem*{corollary-non}{Corollary}

\newcommand{\beqa}{\begin{eqnarray}}
\newcommand{\beq}{\begin{equation}}
\newcommand{\eeqa}{\end{eqnarray}}
\newcommand{\eeq}{\end{equation}}

\newcommand\T[1]{\text{#1}}

\newcommand\imp{\hspace{.2in}\Rightarrow\hspace{.2in}}
\newcommand\lp[3]{\|\text{(#1,#2,#3})\|_p}

\newcommand\comma{\hspace{.2in},\hspace{.2in}}

\newtheorem{innercustomgeneric}{\customgenericname} 
\providecommand{\customgenericname}{}
\newcommand{\newcustomtheorem}[2]{%
  \newenvironment{#1}[1]
  {%
   \renewcommand\customgenericname{#2}%
   \renewcommand\theinnercustomgeneric{##1}%
   \innercustomgeneric
  }
  {\endinnercustomgeneric}
}

\newcustomtheorem{customthm}{Theorem}
\newcustomtheorem{customlemma}{Lemma}

\begin{document}
\title[]{$p$-orderings: From Slater to Kemeny-Young to Ranked Pairs}
\author[]{Amir Babak Aazami and Hubert Lewis Bray}
\address{Clark University\hfill\break\indent Department of Mathematics \hfill\break\indent
Worcester, MA 01610}
\email{aaazami@clarku.edu}
\address{Duke University\hfill\break\indent Department of Mathematics \hfill\break\indent
Durham, NC 27708}
\email{bray@math.duke.edu}

\maketitle
\begin{abstract}
We introduce a family of ranking rules for preferential elections, called \emph{$p$-orderings}, obtained by minimizing the $p$-norm of the pairwise majority margins that disagree with a given ranking. This family is defined on the margin-of-victory matrix of the election and has the Slater orderings as its limit as $p \to 0^+$, includes the Kemeny-Young rule as the case $p=1$, and coincides with Ranked Pairs for all sufficiently large $p$. We show that, under natural assumptions of scale invariance, dependence only on margin magnitude, and monotonicity with respect to margin size, the score function underlying this construction is uniquely of the form $c|x|^p$. Thus Ranked Pairs arises as the eventual large-$p$ member of a canonical family of margin-based ranking rules.
\end{abstract}
\section{Introduction}
In a preferential election, the ballots induce pairwise majority margins between the candidates. Many social welfare rules can be viewed as depending only on this weighted pairwise information, rather than on the full profile itself; see, for example, \cite[Ch.~4]{brandt2016} and, for a broader distance-based perspective, \cite{elkind2015}. In this paper we study a family of such margin-based ranking rules. Given an ordering of the candidates, we consider those pairwise majority margins that disagree with that ordering and measure their size by a $p$-norm. Minimizing this quantity over all rankings defines, for each $p > 0$, a social welfare rule that we call the \emph{$p$-ordering}.
\vskip 6pt
This family connects three familiar voting rules. As $p \to 0^+$, minimizing the $p$-norm becomes equivalent to minimizing the number of pairwise disagreements with the majority relation, and hence selects the Slater orderings \cite{laslier,brandt2016}. The case $p = 1$ recovers the Kemeny-Young rule, and for sufficiently large $p$, the $p$-ordering stabilizes and coincides with Ranked Pairs \cite{tideman, zavist}. Thus the $p$-orderings place the Slater orderings, Kemeny-Young, and Ranked Pairs within a single one-parameter family of rules on pairwise majority margins.
\vskip 6pt
The use of the $p$-norm here is not ad hoc. In Theorem \ref{thm:unique} we show that if one seeks a continuous transformation of pairwise majority margins that is scale-invariant, depends only on margin magnitude, and assigns greater weight to larger margins, then one is led uniquely to functions of the form $c|x|^p$. In this sense, the family of $p$-orderings is canonical among a natural class of margin-based constructions.
\vskip 6pt
Our main theorem is proved on the class of elections for which all pairwise majority margins are nonzero and distinct, so that tie-breaking does not obscure the comparison with Ranked Pairs. As we show in Theorem \ref{thm:RP}, there exists $p_*$, depending on the election, such that every finite $p \geq p_*$ yields exactly the Ranked Pairs ordering. The proof shows that, for large $p$, larger majority margins dominate smaller ones in lexicographic fashion, and the resulting optimization procedure becomes precisely the Ranked Pairs algorithm. We emphasize that the theorem concerns stabilization for sufficiently large $p$; the corresponding sup-norm optimization problem may admit multiple minimizers.

\section{A review of the Ranked Pairs ordering}
\label{sec:RP}
Because our main result identifies Ranked Pairs as the eventual large-$p$ member of the family of $p$-orderings defined in Section \ref{sec:Lp}, in this section we briefly review the Ranked Pairs method. If we are handed three numbers, say 6, 4, and 1, and asked to order them from largest to smallest, then we would promptly write down: $6 > 4 > 1$. Lurking behind this easy task is the property of transitivity: One can never have the absurdity of $6 > 4$ and $4 > 1$ but $1 > 6$. The transitivity of numbers ensures that there is ever only one way to order them from largest to smallest. Unfortunately, we cannot, in a preferential election, order candidates this way. As first observed by M.~Condorcet in the $18^{\text{th}}$ century, when we replace numbers by candidates, and  ``$>$" by ``beats head-to-head," transitivity no longer holds.  In other words, just because candidate A beats candidate B head-to-head (``$\T{A} > \T{B}$"), and B beats C head-to-head (``$\T{B} > \T{C}$"), that does not guarantee that $\T{A} > \T{C}.$ Indeed, it could be that $\T{C} > \T{A}$. If the latter occurs, then that is called a \emph{cycle}. The question therefore arises as to what, if anything, should be done regarding cycles. Tideman's ``Ranked Pairs" voting method \cite{tideman} yields a final ranking by removing the ``weakest link" among certain cycles.  To understand how, and to facilitate our results in Section \ref{sec:limit} below, consider the following election between candidates A, B, C, and D, with all head-to-head matchups represented in a so-called \emph{margin of victory matrix} (m.o.v.~matrix):
\begin{center}
\begin{tabular}{ c|c|c|c|c|l } 
 & A & B & C & D\\
 \cline{1-5}
 A&  & $1$ & $11$ & $-7$ &\footnotesize{\text{$\leftrightsquigarrow$ A beats B by 1, beats C by 11, and loses to D by 7}}\\ 
 \cline{1-5}
 B & $-1$ &  & 5 & $3$ & \footnotesize{\text{$\leftrightsquigarrow$ B loses to A by 1, beats C by 5, and beats D by 3}}\\
 \cline{1-5} 
 C& $-11$ & $-5$ &  & 9 & \footnotesize{\text{$\leftrightsquigarrow$ C loses to A by 11, loses to B by 5, and beats D by 9}}\\
 \cline{1-5} 
 D& $7$ & $-3$ & $-9$ & &\footnotesize{\text{$\leftrightsquigarrow$ D beats A by 7, loses to B by 3, and loses to C by 9}}\\ 
\cline{1-5}
\end{tabular}
\end{center}
\vskip 6pt
Let us begin by listing these head-to-head margins of victory from largest to smallest, like so:
\begin{center}
\begin{tabular}{ c|c } 
\text{margin of victory} & \text{head-to-head outcomes}\\
\hline
11 & A $>$ C\\
\hline
9 & C $>$ D\\
\hline
7 & D $>$ A\\
\hline
5 & B $>$ C\\
\hline
3 & B $>$ D\\
\hline
1 & A $>$ B
\end{tabular}
\end{center}
The method of Ranked Pairs now uses this information to determine a final ordering, by way of the following stipulations: 
\begin{enumerate}[leftmargin=*]
\item[1.)] The margin of victory in head-to-head matchups is given priority.  In Ranked Pairs, A $>$ C will carry more weight than, say, A $>$ B, because in the former the margin of victory was greater.

\item[2.)] Transitivity of head-to-head matchups is obeyed as much as possible, starting from the top of the table. Thus, to determine the final ordering in the election above, write down the first head-to-head matchup that you see,
$$
\text{A} > \text{C},
$$
followed by the second, C $>$ D, in a way that preserves the $>$-ordering already established:
$$
\text{A} > \!\!\underbrace{\,\text{C} > \text{D}\,}_{\text{insert here}}.
$$
Now write down the third, D $>$ A, in the same way, by inserting it where D is situated:
$$
\text{A} > \text{C} > \!\!\underbrace{\,\text{D} > \text{A}\,}_{\text{insert here}}
$$
But this is a cycle, and thus transitivity is violated: We cannot have both ``A $>$ C $>$ D" (and thus ``A $>$ D") and ``D $>$ A" hold at the same time. To rectify this, Ranked Pairs stipulates that we discard D $>$ A because it has the smallest margin of victory in this cycle. We thus remove the head-to-head matchup D $>$ A from our table:

\begin{center}
\begin{tabular}{ c|c } 
\text{margin of victory} & \text{head-to-head outcomes}\\
\hline
11 & A $>$ C\\
\hline
9 & C $>$ D\\
\hline
\xcancel{7} & \xcancel{D $>$ A}~~{\scriptsize (discard)}\\
\hline
5 & B $>$ C\\
\hline
3 & B $>$ D\\
\hline
1 & A $>$ B
\end{tabular}
\end{center}

\item[3.)] Now we continue down the table: Next up is B $>$ C, which is ambiguous, because we can insert it into our current sequence
$$
\text{A} > \text{C} > \text{D}
$$
in two possible ways, each of which respects the ordering already established.  Namely, we can take B $>$ C and insert it into $\text{A} > \text{C} > \text{D}$ like this,
$$
\underbrace{\,\text{A} > \text{B} > \text{C} > \text{D}\,}_{(a)}\,,
$$
or we can  insert it like this:
$$
\underbrace{\,\text{B} > \text{A} > \text{C} > \text{D}\,}_{(b)}.
$$
In both of these cases, B comes \emph{before} C, which is what ``B $>$ C" requires.  (By contrast, we could not have placed B \emph{after} C, as in
$$
\text{A} > \text{C} > \text{B} > \text{D}\hspace{.2in}\text{or}\hspace{.2in}\text{A} > \text{C} > \text{D} > \text{B},
$$
since the relationship ``B $>$ C" is clearly not being reflected in either of these cases.) So, given the possibilities (\emph{a}) and (\emph{b}), Ranked Pairs says that we hold on to both of them for the time being and go on to the next head-to-head matchup in the table, which is B $>$ D.  Notice that this one is already compatible with both (\emph{a}) and (\emph{b})\,---\,B comes before D in both cases\,---\,so there is nothing further to do here and we can move on to the final head-to-head matchup, which is A $>$ B.   This one is compatible only with $(a)$, so we keep the sequence $(a)$ and discard $(b)$:
$$
\text{A} > \text{B} > \text{C} > \text{D}\hspace{.2in}\text{or}\underbrace{\,\xcancel{\text{B} > \text{A} > \text{C} > \text{D}}\,}_{\text{not compatible with $\text{A} > \text{B}$}},
$$
Having processed all head-to-head matchups in decreasing order, and skipping those that would create a cycle (removing the ``weakest link"), we are left with a unique and unambiguous ``winning direction" of head-to-head matchups:
$$
\text{A} > \text{B} > \text{C} > \text{D}.
$$
\item[4.)] Ranked Pairs therefore declares this to be the final ordering, with  A the winner of the election.  What is more, notice that with ``D $>$ A" discarded, the final ordering really is transitive, just like with numbers.
\end{enumerate}

Assuming no ties in margins of victory, the method of Ranked Pairs always yields a unique final ordering. In fact, in that final ordering, each candidate will have necessarily beaten head-to-head the candidate directly below them (though not necessarily all candidates below them).  In addition to this, Ranked Pairs satisfies the following properties:
\begin{enumerate}[leftmargin=*,noitemsep]
\item[1.] \emph{Majority criterion}: A candidate is the winner of the election if they receive over 50\% of first-place votes.
\item[2.] \emph{Condorcet winner criterion}: If a candidate beats all others head-to-head, then they are the winner of the election.
\item[3.] \emph{Condorcet loser criterion}: If a candidate loses to all others head-to-head, then they will not win the election. In Ranked Pairs they will in fact finish last.
\item[4.] \emph{Monotonicity}: The winning candidate cannot fall in the final ordering by becoming more popular, by which is meant one or more voters moving that candidate higher on their ballots (and not changing their relative ordering of the other candidates in the process).
\item[5.] \emph{Clone Invariance}: A group of candidates are \emph{clones}, if for every voter and every candidate not of this group, the voter has the same relative preferences between this candidate and every candidate in the group. Clone invariance says that if clone candidates are added to an election, then either the original winner, or else another member of their clone group, will still win the election.  One consequence of clone invariance is that it dampens the effect of vote splitting.
\item[6.] \emph{Last place loser independence}: Removing the last place loser will not change who wins the election. (The ``last place loser" is defined to be the winner of the election with all ballots reversed.) In fact in Ranked Pairs removing the last place loser will not alter the rest of the ordering at all.
\end{enumerate}

\section{Defining a $p$-norm on the Margin of Victory Matrix}
\label{sec:Lp}
Our goal is to assign a magnitude to pairwise majority margins that disagree with a given ordering. In this section we introduce a way of doing this. The idea is to take the entries of an m.o.v.~matrix that disagree with the given ordering and define a norm on them.  We would like this norm to be sensitive to the ordering of the candidates, because then the ordering with the smallest norm is thereby distinguished. In Section \ref{sec:prop} below we will examine the properties of this norm, known as the ``$p$-norm" in mathematics, but here we will introduce it via the following preferential election with three candidates A, B, C.  Suppose that the m.o.v.~matrix of this election is
$$
\begin{tabular}{ c|c|c|c|l } 
 & A & B & C\\
 \cline{1-4}
 A& & $-1$ & 3 &\footnotesize{\text{$\leftrightsquigarrow$ A lost to B by 1, beats C by 3}}\\ 
 \cline{1-4}
 B & $1$ &  & $-5$ &\footnotesize{\text{$\leftrightsquigarrow$ B beats A by 1, lost to C by 5}}\\
 \cline{1-4} 
 C& $-3$ & $5$ & &\footnotesize{\text{$\leftrightsquigarrow$ C lost to A by 3, beats B by 5}}\\
 \cline{1-4} 
 \end{tabular}
$$
To begin with, the information below the blank diagonal entries, in the lower-left half, is redundant (being determined by the upper right-half), so let us suppress it, as well as the accompanying text, and write simply:
\beqa
\label{eqn:ds}
\begin{tabular}{ c|c|c|c|l } 
 & A & B & C\\
 \cline{1-4}
 A& & $-1$ & $3$\\ 
 \cline{1-4}
 B &  &  & $-5$\\
 \cline{1-4} 
 C&  &  & &\\
 \cline{1-4} 
 \end{tabular}
\eeqa

Throughout the main comparison with Ranked Pairs, we will assume that all the numbers appearing in \eqref{eqn:ds} are distinct and nonzero in any preferential election, to avoid having to implement tie-breaking procedures. With that said, here now is the ``$p$-norm" that we mentioned above:
\begin{enumerate}[leftmargin=*]
\item[1.)] In \eqref{eqn:ds}, disregard any positive numbers and take the absolute value of the negative ones: $|\!-\!1| = 1$ and $|\!-\!5| = 5$. What remains is
$$
\begin{tabular}{ c|c|c|c|l } 
 & A & B & C\\
 \cline{1-4}
 A& & 1 & \\ 
 \cline{1-4}
 B &  &  &5\\
 \cline{1-4} 
 C&  &  & &\\
 \cline{1-4} 
 \end{tabular}
$$
Writing these numbers as $(1,5)$ encourages us to think of this as a point in two-dimensional space, with $x$-coordinate 1 and $y$-coordinate 5. As we'll see below, it would not matter if we had written $(5,1)$ instead.

\item[2.)] Now we take what is called the \emph{$p$-norm} of this point, denoted by ``$\|(1,5)\|_p$."  For any choice of positive real number $p$, this is defined as
$$
\|(1,5)\|_p \defeq (1^p+5^p)^{\frac{1}{p}}.
$$
For example, if we had chosen $p=2$, then the ``$2$-norm of the point $(1,5)$" is the number
$$
\|(1,5)\|_2 = (1^2+5^2)^{\frac{1}{2}} = \sqrt{26} \approx 5.099.
$$
This is, in fact, the familiar Euclidean distance of the point $(1,5)$ from the origin $(0,0)$ in the $xy$-plane (the same distance as the point $(5,1)$). Notice that the choice of $p$ matters.  Indeed, the ``$3$-norm of $(1,5)$" is
$$
\|(1,5)\|_3 = (1^3+5^3)^{\frac{1}{3}} = \sqrt[3]{126} \approx 5.013,
$$
which is smaller than the $2$-norm of $(1,5)$. In Section \ref{sec:prop} we will explore the behavior of the $p$-norm as $p$ increases, but for now, let us pause to record what we have defined: 

\begin{defn}[$p$-norm of an m.o.v.~matrix]
\label{def:Lp}
The \emph{$p$-norm} of an m.o.v.~matrix is the $(1/p)^{\text{th}}$-root of the sum of the \emph{$p^{\text{th}}$} powers of the absolute values of the negative entries in the upper triangular part of the m.o.v.~matrix.
\end{defn}
\noindent Note that, strictly speaking, $\|\cdot\|_p$ is not a norm for $0 < p < 1$, in the sense of mathematical analysis, but we will use ``$p$-norm" for all $p >0$.

\item[3.)] Now we take our final step: Let us consider \emph{all possible permutations} of the candidates A, B, C, and repeat our procedure above for each of them. For example, if instead of the ordering A, B, C in \eqref{eqn:ds} we had instead written the m.o.v.~matrix with respect to the ordering B, C, A, then we would have obtained
$$
\begin{tabular}{ c|c|c|c|l } 
 & B& C & A\\
 \cline{1-4}
 B& & $-5$ & $1$ &\footnotesize{\text{$\leftrightsquigarrow$ B lost to C by 5, beats A by 1}}\\ 
 \cline{1-4}
 C &  &  & $-3$ &\footnotesize{\text{$\leftrightsquigarrow$ C lost to A by 3}}\\
 \cline{1-4} 
 A&  &  & \\
 \cline{1-4} 
 \end{tabular}
$$
Although the ballots have not changed, observe that this is a different matrix than \eqref{eqn:ds}, and thus has a different $p$-norm; e.g., its $2$- and $3$-norms are
\beqa
\|(3,5)\|_2 \!\!&=&\!\! (3^2+5^2)^{\frac{1}{2}} \approx 5.831,\nonumber\\
\|(3,5)\|_3 \!\!&=&\!\! (3^3+5^3)^{\frac{1}{3}} \approx 5.337.\nonumber
\eeqa

\item[4.)] For three candidates, there are $1\cdot 2\cdot 3 = 6$ possible permutations. (For elections with $n$ candidates, there are
$
n! \defeq 1\cdot 2\cdot\cdots \cdot(n-1)\cdot n
$
distinct permutations of the $n$ candidates.) For each of these, we find the $p$-norm of the corresponding m.o.v.~matrix as in Definition \ref{def:Lp}. Notice that the $p$-norms involving only one entry are independent of $p$:
\beqa
\begin{tabular}{ c|c|c|c|l } 
 & A & B & C\\
 \cline{1-4}
 A& & $-1$ & $3$\\ 
 \cline{1-4}
 B &  &  & $-5$\\
 \cline{1-4} 
 C&  &  & &\\
 \cline{1-4} 
 \end{tabular} &\imp& \lp{A}{B}{C} = (1^p+5^p)^{\frac{1}{p}},\label{eqn:KY0}\\
\begin{tabular}{ c|c|c|c|l } 
 & B & C & A\\
 \cline{1-4}
 B& &$-5$& $1$\\ 
 \cline{1-4}
 C &  &  &$-3$ \\
 \cline{1-4} 
 A&  &  & &\\
 \cline{1-4} 
 \end{tabular} &\imp& \lp{B}{C}{A} = (3^p+5^p)^{\frac{1}{p}},\label{eqn:opposite}
 \eeqa
 \beqa
\begin{tabular}{ c|c|c|c|l } 
 & C & A & B\\
 \cline{1-4}
 C& & $-3$ & 5\\ 
 \cline{1-4}
 A &  &  & $-1$\\
 \cline{1-4} 
 B&  &  & &\\
 \cline{1-4} 
 \end{tabular} &\imp& \lp{C}{A}{B} = (1^p+3^p)^{\frac{1}{p}},\label{eqn:cycle*}\\
\begin{tabular}{ c|c|c|c|l } 
 & B & A & C\\
 \cline{1-4}
 B& & $1$ & $-5$\\ 
 \cline{1-4}
 A &  &  & $3$\\
 \cline{1-4} 
 C&  &  & &\\
 \cline{1-4} 
 \end{tabular} &\imp& \lp{B}{A}{C} = (5^p)^{\frac{1}{p}} = 5,\label{eqn:p=1*}\\
\begin{tabular}{ c|c|c|c|l } 
 & A & C & B\\
 \cline{1-4}
 A& & $3$& $-1$\\ 
 \cline{1-4}
 C &  &  & $5$\\
 \cline{1-4} 
 B&  &  & &\\
 \cline{1-4} 
 \end{tabular} &\imp& \lp{A}{C}{B} = (1^p)^{\frac{1}{p}} = 1,\label{eqn:winner}\\
\begin{tabular}{ c|c|c|c|l } 
 & C & B & A\\
 \cline{1-4}
 C& & 5 & $-3$\\ 
 \cline{1-4}
 B &  &  & $1$\\
 \cline{1-4} 
 A&  &  & &\\
 \cline{1-4} 
 \end{tabular} &\imp& \lp{C}{B}{A} = (3^p)^{\frac{1}{p}} = 3.\label{eqn:loser3}
\eeqa
\end{enumerate}

Observe that for every $p > 0$, the ordering with the \emph{smallest} $p$-norm is always \eqref{eqn:winner}, with $\lp{A}{C}{B} = 1$. The $p$-norm has thus distinguished the ordering A $>$ C $>$ B, and thereby defined a new voting method:
\begin{defn}[$p$-ordering]
\label{def:Lp2}
The \emph{$p$-ordering} of the candidates of a preferential election is an ordering of the candidates which minimizes the $p$-norm, as defined in Definition \ref{def:Lp}.
\end{defn}

Minimal $p$-orderings need not be unique in general. However, in Section \ref{sec:limit} we prove that when all pairwise margins of victory are nonzero and distinct, the minimal $p$-ordering is unique for all sufficiently large $p$. Let us observe that we could just as well have taken the largest $p$-norm instead of the smallest, provided we had replaced negative entries with positive ones:

\begin{cor}
\label{cor:vote_2}
If for each m.o.v.~matrix we had taken the $p$-norm of the positive entries instead of the negative ones, then the ordering with the largest $p$-norm is once again the $p$-ordering of Definition \ref{def:Lp2}.
\end{cor}

\begin{proof}
Up to sign and permutation, the numbers in the m.o.v.~matrices above are always the same: If we were to take the $p$-norm of \emph{all} entries, we would always obtain the same sum; e.g., $(1^p+3^p+5^p)^{\frac{1}{p}}$ in the previous example.  Thus, minimizing over just the negative entries is equivalent to maximizing over just the positive ones. This is true more generally for any preferential election. 
\end{proof}

The reader can verify that the largest $p$-norm of the positive entries in the six matrices above is given by \eqref{eqn:winner}, precisely the ordering that gave the $p$-ordering of Definition \ref{def:Lp2}. Therefore, we can just as well work with positive entries instead of negative ones. In any case, we now show that the $p$-norm of Definition \ref{def:Lp} is, in a precise sense, canonical, in that the choice $f(x)=|x|^p$ is not ad hoc. Under the assumptions below, it is the unique continuous score function compatible with this construction. To see how, let us once again consider the $p$-norm of all entries in the upper-right half of an m.o.v.~matrix, this time incorporating the sign of each entry as well, like so:
\beqa
\begin{tabular}{ c|c|c|c|l } 
 & A & B & C\\
 \cline{1-4}
 A& & $-1$ & $3$\\ 
 \cline{1-4}
 B &  &  & $-5$\\
 \cline{1-4} 
 C&  &  & &\\
 \cline{1-4} 
 \end{tabular} \imp Q \defeq -|1|^p+|3|^p-|5|^p.\label{eqn:Q}
\eeqa
We've dropped the exponent $\frac{1}{p}$ as we don't wish to take $p^\text{th}$ roots of negative numbers. In fact \eqref{eqn:Q} can also be used to obtain the $p$-ordering:

\begin{cor}[$Q$-sum]
\label{cor:all}
The ordering whose m.o.v.~matrix has the largest $Q$-sum as in \eqref{eqn:Q} is precisely the $p$-ordering of Definition \ref{def:Lp2}.
\end{cor}

\begin{proof}
The largest value for the $Q$-sum is the m.o.v.~matrix whose negative entries have the smallest $p$-norm in the sense of Definition \ref{def:Lp}.
\end{proof}

Now, if we denote the entry in the $i^{\text{th}}$ row and the $j^{\text{th}}$ column of an m.o.v.~matrix by $m_{ij}$, and let $\mu_{ij}$ denote its sign (i.e., $\mu_{ij} = +1$ for $m_{ij} > 0$ and $\mu_{ij} = -1$ for $m_{ij} < 0$), then $Q$ can be expressed more compactly as
$$
Q = \sum_{i<j} \mu_{ij}|m_{ij}|^p.
$$
Imagine now that we generalize our $p$-ordering by replacing each instance of $|\cdot|^p$ with an arbitrary function $f$, thereby defining a new sum:
$$
Q_f \defeq \sum_{i<j} \mu_{ij}f(m_{ij}).
$$
E.g., $f(m_{ij})$ can be $\text{ln}\,|m_{ij}|$, or $e^{m_{ij}}$, or $\frac{1}{1+m_{ij}^2}$, etc. For each choice of $f$, we would obtain a different $Q_f$-sum and thus possibly a different ordering. Of course, as we are interested in voting and elections, let us make the following assumptions about $f$:
\begin{enumerate}[leftmargin=*]
\item[1.] Orderings should be scale-invariant: For each $a > 0$, there is a \mbox{$b_a>0$} depending on $a$ such that $f(ax) = b_af(x)$ for all $x$. (I.e., whether everyone gets one vote or ten votes shouldn't affect the final ordering.)
\item[2.] Larger-magnitude margins of victory should carry more weight than smaller ones: $f$ is non-decreasing for positive $x$.
\item[3.] Since $\mu_{ij}$ already accounts for the sign, $f$ itself should be a nonnegative even function: $f(-x) = f(x) \geq 0$. (I.e., $f$ should respond only to the magnitude of the margin of victory.)
\end{enumerate}
With these assumptions in place, here is what makes our choice of $f$ unique:

\begin{thm}[Uniqueness of the score function]
\label{thm:unique}
The only continuous functions satisfying assumptions 1-3~above are $f(x) = c|x|^p$ for $c \geq 0,p\geq 0$.
\end{thm}

\begin{proof}
By assumption 1, if $x > 0$, then
$
f(x) = f(x\cdot 1) = b_{x}f(1)
$
for some $b_x>0$. If $x < 0$, then by assumptions 1 and 3,
$$
f(x) = f(|x|) = f(|x|\cdot 1) = b_{|x|}f(1)
$$
for some $b_{|x|}>0$. Therefore $f(x) = b_{|x|}f(1)$ for all $x \neq 0$. By similar reasoning, $f(0) = b_{|x|}f(0)$ for any $x\neq 0$. If any such $b_{|x|} \neq 1$, then we must have $f(0) = 0$; otherwise, $f(x) = f(1)$ for all $x \neq 0$, hence $f(0) =f(1)$ as well, by continuity.  As the constant function $f(x) = f(1)$ is of the form $c|x|^p$ with $c=f(1)$ and $p=0$, this proves the theorem when $f(0) \neq 0$. Let us now assume that $f(0)=0$, and focus on $f(1)$. If $f(1) = 0$, then $f(x) = 0$ for all $x$, which function is again of the form $c|x|^p$ (with $c = 0$ and any $p\geq 0$). What remains is the case when $f(0) = 0$ and $f(1) \neq 0$ (in particular, $f(1) > 0$ by assumption 2). We now finish the proof by showing that here, too, we must have $f(x) = c|x|^p$, this time with $c > 0$ and $p > 0$. To begin with, note that by the same analysis as above,
\beqa
\label{eqn:f}
f(xy) = b_{|x|}f(y) = \frac{f(x)f(y)}{f(1)}
\eeqa
for all $x \neq0$ and $y \in \RR$. Moreover, since $f(0)=0$, since $f(1) > 0$, and since the continuous function $f$ is non-decreasing by assumption 2, we can scale $f$ by a positive number if necessary so that $\int_0^1f(t)\,dt = 1$. Having done so, now define for all $x \geq0$ the new function $F(x) \defeq f(1)\int_0^xf(t)\,dt$. Observe that $F'(x) = f(1)f(x)$, and that
$$
F(x) = \!\!\!\!\!\!\!\!\underbrace{\,f(1)\,x\!\int_0^1f(sx)\,ds\,}_{\text{changing variables to $s =tx^{-1}$}}\!\!\!\!\!\!\! \overset{\eqref{eqn:f}}{=} x\!\int_0^1f(s)f(x)\,ds = xf(x) = \frac{xF'(x)}{f(1)}\cdot
$$
The general solution to this differential equation is $F(x) = cx^{f(1)}$ for $c\in \RR$, from which we extract $f(x) = cx^{f(1)-1}$ upon differentiating $F(x)$; note that $c =f(1) > 0$. In order for $cx^{f(1)-1}$  to be non-decreasing (assumption 2), we must have $f(1) \geq 1$. By assumption 3, $f$ must extend as $c|x|^{f(1)-1}$ for $x < 0$. Putting all of this together, it follows that $f$ is of the form $c|x|^p$, with $c > 0$ and $p \geq 0$.
\end{proof}

Theorem \ref{thm:unique} shows that the $p$-family is canonical within this class of continuous, scale-invariant, magnitude-based score functions.

\section{Properties of the $p$-norm ordering}
\label{sec:prop}
Note that the $p$-norm does not change ``uniformly" as $p>0$ increases, in the following sense: For $x = (10,2)$ and $y=(5,8)$, observe that
$$\|x\|_1 = 12 \comma \|y\|_1 = 13,$$
so that $\|x\|_1 < \|y\|_1$, but 
$$
\|x\|_2 = \sqrt{104} \approx 10.20 \comma \|y\|_2 = \sqrt{89} \approx 9.43,
$$
so that $\|x\|_2 > \|y\|_2$.  Thus, if some m.o.v.~matrix has the lowest $p_*$-norm for some $p_*$, that doesn't guarantee that it will remain the lowest for all $p > p_*$. 

\begin{prop}
\label{prop:Slater}
In the limit $p \to 0^+$, minimizing the $p$-norm is equivalent to finding a Slater ordering.
\end{prop}
\begin{proof}
Fix an ordering $\pi$, and let $x_1,\dots,x_k>0$ be the absolute values of the negative entries in the upper triangular part of its m.o.v.~matrix. Since the function $t \mapsto t^{1/p}$ is strictly increasing for $t>0$, minimizing
\beqa
\label{eqn:S}
\|\pi\|_p=(x_1^p+\cdots+x_k^p)^{\frac{1}{p}}
\eeqa
is equivalent to minimizing the sum $x_1^p+\cdots+x_k^p$, which, unlike \eqref{eqn:S}, has a well defined limit as $p \to 0^+$. For each $i$ we have $x_i^p \to 1$ as $p \to 0^+$, so
$$
x_1^p+\cdots+x_k^p \to k,
$$
the number of pairwise majority comparisons that disagree with $\pi$. Thus, as $p \to 0^+$, the $p$-ordering selects precisely those orderings that minimize the number of pairwise disagreements with the majority relation. These are exactly the Slater orderings.
\end{proof}

Next, the case $p=1$ is also familiar:

\begin{prop}
\label{prop:KY}
The $1$-norm ordering is equivalent to the Kemeny-Young voting method.
\end{prop}

\begin{proof}
Let $\pi = (c_1,\dots,c_n)$ be any ordering of the candidates. The Kemeny-Young score of $\pi$ is the sum of the signed pairwise margins that agree with $\pi$, namely,
$$
\sum_{i<j} m_{c_i c_j}.
$$
When $p=1$, this is exactly the $Q$-sum of Corollary \ref{cor:all}, since $\mu_{ij}|m_{ij}| = m_{ij}$ for every upper-triangular entry of the corresponding m.o.v.~matrix. Hence maximizing the Kemeny-Young score is equivalent to maximizing the $Q$-sum, which by Corollary \ref{cor:all} is equivalent to minimizing the $1$-norm. Therefore the $1$-norm ordering is precisely the Kemeny-Young ordering.
\end{proof}

One consequence of this is that $p$-ordering is not clone invariant in general:

\begin{prop}
The family of $p$-orderings is not clone invariant in general.
\end{prop}

\begin{proof}
The Kemeny-Young method is known to fail clone invariance.
\end{proof}

(Note that when clones are present in a preferential election, the m.o.v.~matrix cannot have all distinct entries.) We will analyze the limit $p \to \infty$ in Section \ref{sec:limit}. For the remainder of this section, we discuss properties that are satisfied by the $p$-ordering. In doing so, the following will be important:

\begin{lemma}[Swapping property]
\label{lemma:swap}
If \emph{A} beats \emph{B} head-to-head, then 
$$
\|(\text{$\dots$\emph{,A,B,}$\,\dots$})\|_p < \|(\text{$\dots$\emph{,B,A,}$\,\dots$})\|_p,
$$
assuming the ordering of the other candidates is the same in each case.
\end{lemma}

\begin{proof}
Denoting by $m_{\scalebox{0.5}{AB}}$ the margin of victory of A over B, and similarly with any other pair of candidates, consider any m.o.v.~matrix in which A and B are adjacent:

{\footnotesize $$
\begin{tabular}{ c|c|c|c|c|c|c } 
 & $\cdots$& A & B & C & D& $\cdots$\\
 \cline{1-7}
 $\vdots$& & $\ddots$ & $\vdots$ &$\vdots$& $\vdots$ & $\vdots$\\ 
 \cline{1-7}
 A& & & $m_{\scalebox{0.5}{AB}}$ & $m_{\scalebox{0.5}{AC}}$ & $m_{\scalebox{0.5}{AD}}$ & $\cdots$\\ 
 \cline{1-7}
 B &  &  & & $m_{\scalebox{0.5}{BC}}$ & $m_{\scalebox{0.5}{BD}}$ & $\cdots$ \\
 \cline{1-7} 
 C&  &  & & & $m_{\scalebox{0.5}{CD}}$ & $\cdots$\\
 \cline{1-7} 
 D&  &  & & &  & $\ddots$ \\
 \cline{1-7}
 $\vdots$&  &  & & & & \\
 \end{tabular}
$$}
\vskip 6pt
If we now swap A and B, then their rows change as follows:
 
 {\footnotesize $$
\begin{tabular}{ c|c|c|c|c|c|c } 
 & $\cdots$& B & A & C & D& $\cdots$\\
 \cline{1-7}
 $\vdots$& & $\ddots$ & $\vdots$ &$\vdots$& $\vdots$ & $\vdots$\\ 
 \cline{1-7}
 B& & & $m_{\scalebox{0.5}{BA}}$ & $m_{\scalebox{0.5}{BC}}$ & $m_{\scalebox{0.5}{BD}}$ & $\cdots$\\ 
 \cline{1-7}
 A &  &  & & $m_{\scalebox{0.5}{AC}}$ & $m_{\scalebox{0.5}{AD}}$ & $\cdots$ \\
 \cline{1-7} 
 C&  &  & & & $m_{\scalebox{0.5}{CD}}$ & $\cdots$\\
 \cline{1-7} 
 D&  &  & & &  & $\ddots$ \\
 \cline{1-7}
 $\vdots$&  &  & & & & \\
 \end{tabular}
$$}

The only upper-triangular entry whose \emph{sign} changes under this adjacent swap is $m_{\scalebox{0.5}{AB}}$, which becomes $m_{\scalebox{0.5}{BA}} = -m_{\scalebox{0.5}{AB}}$; every other upper-triangular entry is merely relabeled. Therefore, if A beats B head-to-head, then $m_{\scalebox{0.5}{AB}} > 0$, so the ordering with A immediately above B does not count this comparison in its $p$-norm, whereas the ordering with B immediately above A counts the negative entry $m_{\scalebox{0.5}{BA}} < 0$. Since all other contributions to the $p$-norm are the same in the two orderings, $\|(\text{$\dots$,A,B,$\,\dots$})\|_p < \|(\text{$\dots$,B,A,$\,\dots$})\|_p.$
\end{proof}

Here are two immediate consequences of Lemma \ref{lemma:swap}.

\begin{prop}
\label{prop:+}
In a preferential election with all pairwise margins of victory distinct and nonzero, every candidate in the $p$-ordering necessarily beats head-to-head the candidate directly below them. In particular, the m.o.v.~matrix yielding the $p$-ordering always has positive entries directly above its diagonal.
\end{prop}

\begin{proof}
This is immediate from Lemma \ref{lemma:swap}, for otherwise we can lower the $p$-norm by swapping some pair of adjacent candidates.
\end{proof}

Proposition \ref{prop:+} greatly narrows the possible candidates for the $p$-ordering. Consider, e.g., the following four-candidate election:

{\footnotesize $$
\underset{\text{$\|(\text{C,A,B,D})\|_p = (5^p+3^p)^{\frac{1}{p}}$}}{\begin{tabular}{ c|c|c|c|c|l } 
 & C & A & B & D\\
 \cline{1-5}
 C& & 1 & $-5$ & 9\\ 
 \cline{1-5}
 A &  &  & 11 & $-3$\\
 \cline{1-5} 
 B&  &  & & 7\\
 \cline{1-5} 
 D&  &  & &\\
 \cline{1-5}
 \end{tabular}} \comma
\underset{\text{$\|(\text{C,D,A,B})\|_p = (5^p+7^p)^{\frac{1}{p}}$}}{\begin{tabular}{ c|c|c|c|c|l } 
 & C & D & A & B\\
 \cline{1-5}
 C& & 9 & 1 & $-5$\\ 
 \cline{1-5}
 D &  &  & 3 & $-7$\\
 \cline{1-5} 
 A&  &  & & 11\\
 \cline{1-5} 
 B&  &  & &\\
 \cline{1-5}
 \end{tabular}} \comma
 \underset{\text{$\|(\text{A,B,C,D})\|_p = (1^p+3^p)^{\frac{1}{p}}$}}{\begin{tabular}{ c|c|c|c|c|l } 
 & A & B & C & D\\
 \cline{1-5}
 A& & 11 & $-1$ & $-3$\\ 
 \cline{1-5}
 B &  &  & 5 & 7\\
 \cline{1-5} 
 C&  &  & & 9\\
 \cline{1-5} 
 D&  &  & &\\
 \cline{1-5}
 \end{tabular}}
$$}

{\footnotesize \beqa
\underset{\text{$\|(\text{B,C,D,A})\|_p = (11^p)^{\frac{1}{p}} = 11$}}{\begin{tabular}{ c|c|c|c|c|l } 
 & B & C  &D  &A \\
 \cline{1-5}
 B& & 5 & 7 & $-11$\\ 
 \cline{1-5}
 C &  &  & 9 & 1\\
 \cline{1-5} 
 D&  &  & & 3\\
 \cline{1-5} 
 A&  &  & &\\
 \cline{1-5}
 \end{tabular}}\ \comma\!
 \underset{\text{$\|(\text{D,A,B,C})\|_p = (1^p+7^p+9^p)^{\frac{1}{p}}$}}{\begin{tabular}{ c|c|c|c|c|l } 
 & D & A & B  &C \\
 \cline{1-5}
 D& & 3 & $-7$ & $-9$\\ 
 \cline{1-5}
  A&  &  & 11 & $-1$\\
 \cline{1-5} 
  B&  &  & & 5\\
 \cline{1-5} 
  C&  &  & &\\
 \cline{1-5}
 \end{tabular}\label{eqn:four}}
\eeqa}

Among the $4!=24$ possible orderings, these five are the only ones whose entries immediately above the diagonal are positive; among them, the smallest $p$-norm is achieved by A $>$ B $>$ C $>$ D. The second consequence of Lemma \ref{lemma:swap} is that it always places Condorcet winners at the top and Condorcet losers at the bottom; more generally, $p$-orderings are Smith-consistent:

\begin{prop}[Smith-consistency]
\label{prop:Cond}
For all $p > 0$, $p$-ordering is Smith-consistent. That is, if the candidate set partitions as $\mathcal{A} \sqcup \mathcal{B}$, with each candidate in $\mathcal{A}$ beating head-to-head every candidate in $\mathcal{B}$, then every $p$-ordering ranks every candidate in $\mathcal{A}$ above every candidate in $\mathcal{B}$. 
\end{prop}
\begin{proof}
Suppose, toward a contradiction, that some $p$-ordering contains some candidate in $\mathcal{B}$ above some candidate in $\mathcal{A}$. Then there must be an adjacent occurrence of the form B $>$ A, for some B $\in \mathcal{B}$ and A $\in \mathcal{A}$.  But as A beat B head-to-head, by Lemma \ref{lemma:swap} the $p$-norm will decrease upon swapping A and B, a contradiction.
\end{proof}

Note that all Condorcet-compatible methods fail the \emph{participation criterion}, as shown in \cite{moulin} (see also \cite{fishburn,perez,saari}). For weaker variants of this property (e.g., \emph{positive} or \emph{negative involvement}), which can be categorized as ``no-show paradoxes," as well as for the compatibility of Condorcet-consistent voting methods with ``spoiler effects," see, e.g., \cite{wes}. The next two properties we prove, last place loser independence and monotonicity, both rest on the following ``removal" property of the $p$-ordering:

\begin{lemma}[Removing the first or last candidate]
\label{lemma:lpl}
If the first- or last-placed candidate of a $p$-ordering is removed, then the ordering left over is a $p$-ordering of the election without that candidate.
\end{lemma}
\begin{proof}
Let $E$ be the original election, and let $E \setminus \{A\}$ denote the election obtained by removing a candidate $A$. Suppose that $A$ is first in a $p$-ordering of $E$. For any ordering $\pi$ of the remaining candidates, the comparisons involving $A$ contribute the same quantity to the $p$-norm of $(A,\pi)$, namely,
$$
\sum_{\substack{B \neq A \\ \text{$m_{\scalebox{0.5}{AB}}<0$}}} |m_{\scalebox{0.5}{AB}}|^p.
$$
Therefore
$$
\big(\|(A,\pi)\|_p\big)^p = \!\!\!\sum_{\substack{B \neq A \\ \text{$m_{\scalebox{0.5}{AB}}<0$}}} |m_{\scalebox{0.5}{AB}}|^p + (\|\pi\|_{p,\scalebox{0.5}{$E\setminus\{A\}$}})^p,
$$
where $\|\pi\|_{p,\scalebox{0.5}{$E\setminus\{A\}$}}$ denotes the $p$-norm computed in the reduced election \mbox{$E \setminus \{A\}$.} Since the first term is independent of $\pi$, minimizing $\|(A,\pi)\|_p$ over all orderings with $A$ first is equivalent to minimizing $\|\pi\|_{p,\scalebox{0.5}{$E\setminus\{A\}$}}$. Hence the ordering left after removing $A$ is precisely the $p$-ordering of $E \setminus \{A\}$. The argument when $A$ is last is identical. For any ordering $\pi$ of the remaining candidates,
$$
\big(\|(\pi,A)\|_p\big)^p = (\|\pi\|_{p,\scalebox{0.5}{$E\setminus\{A\}$}})^p + \!\!\!\sum_{\substack{B \neq A \\ \text{$m_{\scalebox{0.5}{BA}}<0$}}} |m_{\scalebox{0.5}{BA}}|^p,
$$
so again the contribution involving $A$ is constant, and minimizing over orderings with $A$ last is equivalent to minimizing the $p$-norm in the reduced election. This proves the lemma.
\end{proof}
\begin{prop}[Last Place Loser Independence]
For all $p > 0$, the $p$-ordering of the reversed-ballot election is precisely the reverse of the original one. As a consequence, $p$-ordering is last place loser independent.
\end{prop}
\begin{proof}
Let $E^{\mathrm{rev}}$ denote the election obtained by reversing every ballot of the original election $E$, and let $\pi = (c_1,\dots,c_n)$ be any ordering of the candidates. Write $\pi^{\mathrm{rev}} = (c_n,\dots,c_1)$ for the reverse ordering. If $\overline{m}_{ij}$ denotes the pairwise margins in $E^{\mathrm{rev}}$, then $\overline{m}_{ij} = -m_{ij}$ for all $i,j$. Hence, for every $i<j$,
$$
\overline{m}_{c_jc_i} = -m_{c_jc_i} = m_{c_ic_j}.
$$
Therefore the upper-triangular entries of the m.o.v.~matrix of $E^{\mathrm{rev}}$ with respect to $\pi^{\mathrm{rev}}$ are exactly the upper-triangular entries of the m.o.v.~matrix of $E$ with respect to $\pi$. In particular,
$$
Q_{\scalebox{0.6}{$E^{\mathrm{rev}}$}}(\pi^{\mathrm{rev}}) = Q_{\scalebox{0.6}{$E$}}(\pi).
$$
It follows that $\pi$ maximizes $Q_{\scalebox{0.6}{$E$}}$ if and only if $\pi^{\mathrm{rev}}$ maximizes $Q_{\scalebox{0.6}{$E^{\mathrm{rev}}$}}$. By Corollary \ref{cor:all}, the $p$-ordering of the reversed-ballot election is therefore the reverse of the original $p$-ordering. Now let $L$ be the last place loser of the original election, i.e., the winner of the reversed-ballot election. By the first part of the proposition, $L$ is exactly the candidate appearing last in the original $p$-ordering. Lemma \ref{lemma:lpl} therefore implies that removing $L$ leaves the $p$-ordering of the remaining candidates unchanged. In particular, the winner is unchanged, so $p$-ordering is last place loser independent.
\end{proof}
\begin{prop}[Monotonicity]
\label{prop:mon}
For any $p > 0$, if the $p$-ordering is unique, then it satisfies the monotonicity property.
\end{prop}
\begin{proof}
Let $E$ be the original election, and let $E'$ be the election obtained by moving the winning candidate $A$ higher on one or more ballots, without changing the relative ordering of any other pair of candidates. Let $\pi$ be the $p$-ordering of $E$, with $A$ in first place. For each candidate $B \neq A$, the pairwise margin of $A$ over $B$ is nondecreasing when passing from $E$ to $E'$. Define
$$
\phi(x) \defeq \operatorname{sgn}(x)|x|^p \qquad \text{and} \qquad \delta_{\scalebox{0.5}{$B$}} \defeq \phi(m'_{\scalebox{0.5}{AB}}) - \phi(m_{\scalebox{0.5}{AB}}),
$$
where $m_{\scalebox{0.5}{AB}}$ and $m'_{\scalebox{0.5}{AB}}$ are the pairwise margins of $A$ over $B$ in $E$ and $E'$, respectively. Since $\phi$ is increasing and $m'_{\scalebox{0.5}{AB}} \geq m_{\scalebox{0.5}{AB}}$, we have $\delta_{\scalebox{0.5}{$B$}} \geq 0$ for every $B \neq A$. Now fix any ordering $\rho$ of the candidates. The change in the $Q$-sum from $E$ to $E'$ comes only from comparisons involving $A$. If $B$ lies below $A$ in $\rho$, then the comparison between $A$ and $B$ contributes $+\delta_{\scalebox{0.5}{$B$}}$ to $Q_{\scalebox{0.6}{$E'$}}(\rho)-Q_E(\rho)$; if $B$ lies above $A$ in $\rho$, then it contributes $-\delta_{\scalebox{0.5}{$B$}}$. Therefore
$$
Q_{\scalebox{0.6}{$E'$}}(\rho)-Q_{\scalebox{0.6}{$E$}}(\rho)
\ = \hspace{-.25in}\sum_{\substack{B \neq A \\ B\text{ below }A\text{ in }\rho}} \hspace{-.25in}\delta_{\scalebox{0.5}{$B$}}
\ \ \ -\!\!\! \sum_{\substack{B \neq A \\ B\text{ above }A\text{ in }\rho}} \hspace{-.25in}\delta_{\scalebox{0.5}{$B$}}
\leq \sum_{B \neq A}\delta_{\scalebox{0.5}{$B$}},
$$
with equality precisely when $A$ is first in $\rho$. Since $A$ is first in $\pi$, we obtain
$$
Q_{\scalebox{0.6}{$E'$}}(\pi) = Q_{\scalebox{0.6}{$E$}}(\pi) + \sum_{B \neq A} \delta_{\scalebox{0.5}{$B$}},
$$
from which it follows that
$$
Q_{\scalebox{0.6}{$E'$}}(\rho) \leq Q_{\scalebox{0.6}{$E$}}(\rho)\ + \sum_{B \neq A} \delta_{\scalebox{0.5}{$B$}} \leq Q_{\scalebox{0.6}{$E$}}(\pi)\ + \sum_{B \neq A} \delta_{\scalebox{0.5}{$B$}} = Q_{\scalebox{0.6}{$E'$}}(\pi),
$$
because $\pi$ maximizes $Q_{\scalebox{0.6}{$E$}}$ by Corollary \ref{cor:all}. Thus $\pi$ also maximizes $Q_{\scalebox{0.6}{$E'$}}$, so $A$ remains first in a $p$-ordering after being moved higher on the ballots. This is exactly the monotonicity property.
\end{proof}

A final property of the $p$-ordering has to do with the absence of cycles:

\begin{prop}
For a given preferential election with all pairwise margins of victory distinct and nonzero, if there are no cycles, then there is a unique ordering of the candidates so that the corresponding m.o.v.~matrix has zero $p$-norm for all $p > 0$. Consequently, in such a case all $p$-orderings will yield the same final ordering.
\end{prop}

\begin{proof}
If there are no cycles, then the head-to-head relation is transitive. Therefore there is a unique ordering $(c_1,\dots,c_n)$ such that $c_i$ beats $c_j$ whenever $i<j$. With respect to this ordering, every upper-triangular entry of the m.o.v.~matrix is positive, so the corresponding $p$-norm is zero for every $p>0$. Conversely, if an ordering has zero $p$-norm, then every upper-triangular entry of its m.o.v.~matrix must be positive. Thus each earlier candidate beats each later candidate head-to-head, so this ordering realizes the transitive majority relation. By uniqueness of the transitive ordering, it must be $(c_1,\dots,c_n)$. Hence there is a unique ordering with zero $p$-norm, and consequently all $p$-orderings coincide with it.
\end{proof}

\section{Ranked Pairs and the convergence of the $p$-norm}
\label{sec:limit}
We now present our main result. First, given an m.o.v.~matrix

{\footnotesize $$
\begin{tabular}{ c|c|c|c|c|c|c } 
 & A& B & C & D & E& $\cdots$\\
 \cline{1-7}
 A& & $m_{12}$ & $m_{13}$ &$m_{14}$& $m_{15}$ & $\cdots$\\ 
 \cline{1-7}
 B& & & $m_{23}$ &$m_{24}$& $m_{25}$ & $\cdots$\\ 
 \cline{1-7}
 C &  &  & & $m_{34}$ & $m_{35}$ & $\cdots$ \\
 \cline{1-7} 
 D&  &  & & & $m_{45}$ & $\cdots$\\
 \cline{1-7} 
 E&  &  & & &  & $\cdots$ \\
 \cline{1-7}
 $\vdots$&  &  & & & & \\
 \end{tabular}
$$}

recall from Corollary \ref{cor:all} that the $p$-ordering is the one with the largest $Q$-sum, where the latter is defined by
$$
Q = \pm |m_{12}|^p \pm |m_{13}|^p \pm \cdots \pm |m_{n-1,n}|^p.
$$
Crucial to our result is the following property of the $Q$-sum:

\begin{lemma}[Cumulative Dominance Property]
\label{lemma:CDP}
For any preferential election with all pairwise margins of victory distinct, there is a value $p_*$, depending on the election, such that for all $p \geq p_*$ each term in the $Q$-sum exceeds the sum of all smaller terms:
\beqa
\label{eqn:CDP}
|m_{ij}|^p\ > \hspace{-.3in}\underbrace{\,\sum |m_{ab}|^p\,}_{\text{sum over all $|m_{ab}| < |m_{ij}|$}}\hspace{.2in}\text{for each $|m_{ij}|$ and for all $p \geq p_*$}.
\eeqa
\end{lemma}

\begin{proof}
For any $|m_{ij}|$, there are fewer than ${n \choose 2} = \frac{n(n-1)}{2}$ smaller terms, because ${n \choose 2}$ is the total number of pairwise matchups possible when there are $n$ candidates. Therefore, as each smaller term $|m_{ab}| \leq |m_{ij}|-1$, it is enough to show that there is a value $p_*$ satisfying
\beqa
\label{eqn:ineq}
|m_{ij}|^p >  \frac{n(n-1)}{2}\big((|m_{ij}|-1)^p\big) \hspace{.2in}\text{for all $p \geq p_*$}.
\eeqa
Assuming that $m_{ij}$ is not the smallest term, so that $|m_{ij}|-1 >0$, \eqref{eqn:ineq} is equivalent to
\mbox{$p > \frac{\text{ln}\big(\frac{n(n-1)}{2}\big)}{\text{ln}\big(\frac{|m_{ij}|}{|m_{ij}|-1}\big)}\cdot$}
Because the election has only finitely many pairwise margins, we may choose $p_*$ to be the maximum of these finitely many bounds over all $|m_{ij}|$ that are not minimal.
\end{proof}

To illustrate this property, consider our four-candidate election \eqref{eqn:four}, for which the $Q$-sum of the ordering B $>$ C $>$ D $>$ A is
$$
Q_{\scalebox{0.4}{B $>$ C $>$ D $>$ A}} = -11^p + 9^p + 7^p + 5^p + 3^p + 1^p.
$$
Observe that, while $11^p < 9^p+9^p+9^p+9^p+9^p$ for all $p \leq 8$, $11^p$ dominates for all $p \geq 9$ (because the ratio $\big(\frac{11}{9}\big)^{9} > 1+1+1+1+1$). It follows that $11^p > 9^p+7^p+5^p+3^p+1^p$ for all $p \geq 9$ as well.

\begin{thm}
\label{thm:RP}
For all $p \geq p_*$ as in Lemma \ref{lemma:CDP}, the Ranked Pairs ordering uniquely maximizes the $Q$-sum, and therefore coincides with the $p$-ordering.
\end{thm}

\begin{proof}
Let $m_1, m_2,\dots,m_{{n\choose 2}}$ denote the margins of victory ordered from largest to smallest, all distinct and nonzero. We will repeatedly make use of the fact that, by \eqref{eqn:CDP}, for $p \geq p_*$ any $|m_i|^p$ will dominate the sum of all smaller terms in the $Q$-sum after it. To maximize the $Q$-sum, we therefore begin by restricting to orderings with both $m_1$ and $m_2$ positive; note that this can always be done. If we can also make $m_3$ positive, then we must do so because by \eqref{eqn:CDP} it will dominate the sum of all smaller terms. However, it is possible that a cycle may prohibit an ordering in which $m_1,m_2$, and $m_3$ are all positive. If that is the case, then make $m_3$ negative. Repeat this process, making every subsequent $m_i$ positive if it is possible to do so without changing the signs of $m_1,m_2,\dots,m_{i-1}$, and negative otherwise. Appealing to \eqref{eqn:CDP} at every step ensures that this ordering will uniquely maximize the $Q$-sum for $p \geq p_*$. But what we have just described is precisely the Ranked Pairs algorithm.
\end{proof}

\begin{cor}
\label{cor:sup}
The Ranked Pairs ordering minimizes the sup-norm $\|\cdot\|_{\scalebox{0.6}{$\infty$}}$ of the reduced m.o.v.~matrix. However, the orderings attaining the minimal sup-norm need not be unique.
\end{cor}
\begin{proof}
First, recall that the sup-norm $\|(x_1,\dots,x_n)\|_{\scalebox{0.6}{$\infty$}} = \text{max}\{|x_1|,\dots,|x_n|\}$ is formally the limit of the $p$-norm as $p \to \infty$. Now, let $\pi_{\scalebox{0.5}{RP}}$ denote the Ranked Pairs ordering, and suppose, for the sake of contradiction, that some ordering $\rho$ satisfies
$\|\rho\|_{\scalebox{0.6}{$\infty$}} < \|\pi_{\scalebox{0.5}{RP}}\|_{\scalebox{0.6}{$\infty$}}.$
Since $\|\pi_{\scalebox{0.5}{RP}}\|_{p} \to \|\pi_{\scalebox{0.5}{RP}}\|_{\scalebox{0.6}{$\infty$}}$ and $\|\rho\|_p \to \|\rho\|_{\scalebox{0.6}{$\infty$}}$, it follows that for all $p \geq M$ for some $M$, we would also have
$$
\|\rho\|_p < \|\pi_{\scalebox{0.5}{RP}}\|_p.
$$
But for $p \geq \text{max}\{p_*,M\}$, this contradicts Theorem \ref{thm:RP}. Therefore $\pi_{\scalebox{0.5}{RP}}$ must minimize the sup-norm. The second statement follows because distinct orderings may have the same largest violated margin.
\end{proof}

Taken together, Proposition \ref{prop:Slater}, Proposition \ref{prop:KY}, and Theorem \ref{thm:RP} show that the family of $p$-orderings (Definition \ref{def:Lp2}) interpolates between three familiar ways of measuring disagreement with the pairwise majority relation: counting disagreements (Slater) as $p \to 0^+$, summing their magnitudes (Kemeny--Young) at $p=1$, and eventually prioritizing larger disagreements lexicographically (Ranked Pairs) for all sufficiently large $p$. Regarding this $p$-spectrum, let us close by showing that there are only finitely many intervals on the positive $p$-axis on which the set of $p$-orderings is constant.
\begin{prop}[Finiteness of the $p$-spectrum]
\label{prop:spec}
For a fixed election with nonzero pairwise margins, the positive $p$-axis is divided into finitely many intervals on which the set of $p$-orderings is constant.
\end{prop}
\begin{proof}
Fix an ordering $\pi$, and let $x_1,\dots,x_k>0$ be the absolute values of the negative entries in the upper triangular part of its m.o.v.~matrix.  Recall from \eqref{eqn:S} that minimizing the $p$-norm $\|\pi\|_p$ is equivalent to minimizing
$$
S_{\pi}(p) \defeq (\|\pi\|_p)^p = \sum_{i=1}^k |x_i|^p = \sum_{i=1}^k e^{p\,\text{ln}|x_i|}.
$$
For a fixed election, the only way for the $p$-ordering to change as $p$ changes is if $S_{\pi}(p_*) = S_{\rho}(p_*)$ for two orderings $\pi,\rho$ and some $p_*$. Therefore, take two orderings $\pi,\rho$ and consider the difference $S_{\pi}(p) - S_{\rho}(p)$; after combining any equal exponents, this difference takes the general form
$$
S_{\pi}(p) - S_{\rho}(p) = \sum_{i=1}^r a_ie^{p\,\text{ln}|x_i|},
$$
for nonzero $a_i \in \mathbb{Z}$. As a function of $p>0$, this difference is either identically zero or else has finitely many zeros (the latter follows by induction on the number of terms $e^{p\,\text{ln}|x_i|}$, and the fact that if $S_{\pi}(p)$ has infinitely many zeros, then so does its derivative $S_{\pi}'(p)$, by the Mean Value Theorem). But there are only finitely many pairs $\pi,\rho$ of orderings, and pairs for which the function $S_{\pi}(p) - S_{\rho}(p)$ is identically zero do not create transition values. For all other pairs there are only finitely many values at which $S_{\pi}(p) = S_{\rho}(p)$, and hence there can only be finitely many intervals on which the set of $p$-orderings is constant.
\end{proof}

\bibliographystyle{alpha}
{\footnotesize\bibliography{p_norm}}


\end{document}